\documentclass[aps,pra,superscriptaddress,nofootinbib,twocolumn,floatfix]{revtex4}
\usepackage{color}
\usepackage{hyperref}
\usepackage[T1]{fontenc} 
\usepackage[latin1]{inputenc} 
\usepackage{graphicx,color} 
\usepackage{amssymb,amsmath,amsthm} 
\usepackage{float} 
\usepackage{appendix}
\DeclareGraphicsExtensions{.jpg}
\providecommand{\openone}{\leavevmode\hbox{\small1\kern-3.8pt\normalsize1}} 

\newcommand{\half}{\frac{1}{2}}

\definecolor{nblue}{rgb}{0.2,0.2,0.7}
\definecolor{ngreen}{rgb}{0.2,0.6,0.2}
\definecolor{nred}{rgb}{0.8,0.2,0.2}
\definecolor{nblack}{rgb}{0,0,0}

\newcommand{\blk}{\color{nblack}}

\newcommand{\unit}{\mathbf{1}}

\renewcommand{\mod}[2]{\left[#1\right]_{#2}}
\renewcommand{\L}{{L}}
\renewcommand{\S}{{S}}
\newcommand{\Q}{{Q}}
\newcommand{\B}{{B}}
\newcommand{\bounds}[4]{\beta^{#1}_{#2,#3,#4}}
\newcommand{\bound}[1]{\beta^{#1}_{n,m,k}}
\newcommand{\boundsf}[5]{\beta^{#1}_{#2,#3,#4;\, #5}}
\newcommand{\boundf}[2]{\beta^{#1}_{n,m,k;\, #2}}

\newcommand{\GUBI}[3]{{\cal I}_{#1,#2,#3}}
\newcommand{\GUBIg}{\GUBI{n}{m}{k}}
\newcommand{\GGUBI}[3]{\Omega_{#1,#2,#3;\, f}}
\newcommand{\GGUBIf}[4]{\Omega_{#1,#2,#3;\, #4}}
\newcommand{\GGUBIg}{\GGUBI{n}{m}{k}}

\newcommand{\sums}{\mathbf{s}} 
\newcommand{\sumr}{\mathbf{r}_{\vec{s}}} 
\newcommand{\sumrP}{\mathbf{r}_{\vec{s}'}'} 

\newcommand{\f}{f}
\newcommand{\g}{g}
\newcommand{\s}{s}
\renewcommand{\r}{r}

\newtheorem{lemma}{Lemma}

\begin{document}
\title{A framework for the study of symmetric full-correlation Bell-like inequalities}

\author{Jean-Daniel Bancal}
\affiliation{Group of Applied Physics, University of Geneva, CH-1211 Geneva 4, Switzerland}

\author{Cyril Branciard}
\affiliation{School of Mathematics and Physics, The University of Queensland, St Lucia, QLD 4072, Australia}

\author{Nicolas Brunner}
\affiliation{H.H. Wills Physics Laboratory, University of Bristol, Tyndall Avenue, Bristol, BS8 1TL, United Kingdom}

\author{Nicolas Gisin}
\affiliation{Group of Applied Physics, University of Geneva, CH-1211 Geneva 4, Switzerland}

\author{Yeong-Cherng Liang}
\affiliation{Group of Applied Physics, University of Geneva, CH-1211 Geneva 4, Switzerland}

\date{\today}

\begin{abstract}
Full-correlation Bell-like inequalities represent an important subclass of Bell-like inequalities that have found applications in both a better understanding of fundamental physics and in quantum information science. Loosely speaking, these are inequalities where only measurement statistics involving all parties play a role. In this paper, we provide a framework for the study of a large family of such inequalities that are symmetrical with respect to arbitrary permutations of the parties.  As an illustration of the power of our framework, we derive $(i)$ a new family of Svetlichny inequalities for arbitrary numbers of parties, settings and outcomes, $(ii)$ a new family of two-outcome device-independent entanglement witnesses for genuine $n$-partite entanglement and $(iii)$ a new family of two-outcome Tsirelson inequalities for arbitrary numbers of parties and settings. We also discuss briefly the application of these new inequalities in the characterization of quantum correlations.
\end{abstract}

\maketitle

\section{Introduction}
\label{Sec:Intro}

Bell inequalities~\cite{bell} play a central role in quantum physics and in quantum information~\cite{R.F.Werner:QIC:2001}. Initially discovered in the context of foundational research on quantum correlations, they are today used in a wide range of protocols for quantum information processing. For instance, they are naturally associated with communication complexity~\cite{CC}, and are the key ingredient in device-independent quantum information processing~\cite{Mayers_Yao,DIQKD_PRL,rand_pironio,rand_colbeck,DISE,EntMeas,diew}. Thus, developing and harnessing Bell inequalities is fundamental towards a deeper understanding of the foundations of quantum mechanics, as well as for applications in quantum information.

The most famous and widely-used Bell inequality is due to Clauser-Horne-Shimony-Holt (CHSH)~\cite{Bell:CHSH}. The CHSH scenario, which is the simplest nontrivial Bell scenario, involves two parties each performing two possible binary-outcome measurements. Denoting by $x$ and $y$ the measurement settings of Alice and Bob respectively, and by $A_x$ and $B_y$ their measurement outcomes, the CHSH inequality reads $E_{11}+E_{12}+E_{21}-E_{22} \leq 2$, where the two-party correlators $E_{xy}$ are defined as $E_{xy}=P(A_x{=}B_y)-P(A_x{\neq}B_y)$\footnote{Throughout the paper, the notation $P$ denotes probabilities.}. Clearly, the value of the correlator does not depend on the individual values of Alice's and Bob's outcomes, but rather on how $A_x$ and $B_y$ relate to each other. Since the CHSH inequality is expressed in terms of these correlators only, it is said to be a correlation Bell inequality.

It is natural and useful to investigate Bell tests beyond CHSH. Bell scenarios can indeed involve in general an arbitrary number of parties, each party  having an arbitrary number of measurement settings, and each of the corresponding measurements an arbitrary number of possible outcomes. Here we denote by the triple $(n,m,k)$ a Bell scenario where $n$ parties all have $m$ possible measurement settings with $k$ possible outcomes. Correlation Bell inequalities can also be naturally  defined in these situations and represent powerful tools for investigating nonlocality (see, e.g. Refs.~\cite{WW:2001,Hoban:2011}). 

In this regard, we will refer to a $k$-valued function of all parties' measurement outcomes as a {\em full-correlation function} if the function can still take on all $k$ possible values even when all but one of the parties' outcomes (for given measurement settings) are fixed. A full-correlation Bell inequality is then one that can be written as a linear combination of probabilities associated with full-correlation function taking particular values. These inequalities are natural generalizations of the Bell-correlation inequalities considered by Werner and Wolf in Ref.~\cite{WW:2001} to arbitrary number of measurement outcomes. In this paper, we shall consider specifically inequalities where the full-correlation function involved is the sum (modulo $k$) of all parties' measurement outcomes. 

Up until now, several families of full-correlation Bell inequalities have been discovered for specific cases. First, for the multi-input $(2,m,2)$ case, Pearle, followed by Braunstein and Caves, introduced the chained Bell inequalities~\cite{chainedBI}. In the multipartite $(n,2,2)$ case, the CHSH inequality has then been generalized by Mermin and further developed by Ardehali, Belinski\v{\i} and Klyshko (MABK)~\cite{MerminIneq}. In fact, a complete characterization of all the $2^{2^n}$ full-correlation Bell inequalities present in this scenario was later achieved by Werner and Wolf~\cite{WW:2001}, and independently by \.Zukowski and Brukner~\cite{ZB:2002}. For the $(2,2,k)$ case, Collins-Gisin-Linden-Massar-Popescu (CGLMP) derived correlation inequalities for scenarios with arbitrary number of measurement outcomes~\cite{CGLMP} (see also Ref.~\cite{Kaslikowski}). Finally, Barrett-Kent-Pironio (BKP) presented in Ref.~\cite{BKP} Bell inequalities for the $(2,m,k)$ case, unifying the CGLMP and the chained Bell inequalities.

Beyond standard Bell inequalities, other types of inequalities are worth considering. These include Tsirelson inequalities~\cite{B.S.Tsirelson:LMP:1980}, which are satisfied by all quantum correlations; Svetlichny inequalities~\cite{svet87}, which can be used to detect genuine multipartite nonlocality; and device-independent entanglement witnesses (DIEWs)~\cite{diew,diewTV}, which detects genuine multipartite entanglement even with Svetlichny-local correlations. We shall refer to all these inequalities as Bell-like inequalities.

For detecting genuine multipartite nonlocality, Collins {\em et al.}~\cite{collins02} and Seevinck-Svetlichny~\cite{Seevinck} have derived full-correlation Svetlichny inequalities for the $(n,2,2)$ case, generalizing those in Ref.~\cite{svet87}. Just like the BKP Bell inequalities, which can be seen as generalization of the chained Bell inequalities to more outcomes, a generalization of the Svetlichny inequalities~\cite{collins02,Seevinck} to the $(n,2,k)$ scenario was also achieved in Ref.~\cite{BBGL}, effectively unifying the CGLMP inequality and the generalized Svetlichny inequalities of Refs.~\cite{collins02,Seevinck}. 

Since all the aforementioned families of Bell-like inequalities reduce to CHSH for $n{=}m{=}k{=}2$, a natural question that one may ask is whether it is possible to unify all these inequalities into a single family of mathematical expression (henceforth referred as Bell expression) for the general $(n,m,k)$ scenario (see Fig.~\ref{Fig:GUBI}). In this paper, we provide an affirmative answer to this question.

To achieve this, we will start in Sec.~\ref{Sec:Framework} by presenting a unified Bell  expression  that, together with the appropriate bound,  reduces to all the Bell-like correlation inequalities mentioned in the last paragraphs as limiting cases. This effectively provides a unified framework for the study of a large family of full-correlation Bell-like inequalities. After that, in Sec.~\ref{Sec:Construction}, we discuss how new multipartite Tsirelson inequalities, Svetlichny inequalities and DIEWs can be constructed within our framework, starting from the respective bipartite and tripartite bounds. Explicit examples of such Bell-like inequalities are then presented. We then conclude in Sec.~\ref{Sec:Conclusion} with some possible avenues for future research.

\section{A framework for symmetric full-correlation Bell-like inequalities}
\label{Sec:Framework}

\subsection{A unified Bell expression  }
\label{Sec:GUBI}

In this section, we present a unified Bell  expression   that reduces to various known Bell  expressions   as special cases. For definiteness, let us label the measurement settings (inputs) for the $i$-th party as $s_i=0,1,\ldots,m-1$ and denoted the corresponding outcome by $r_{s_i}=0,1,\ldots,k-1$. For convenience, we will also write $\vec{s}=(s_1,s_2,\ldots,s_n)$, 
and define the sums of all parties' inputs and outputs, respectively, as $\sums=\sum_{i=1}^n s_i$ and $\sumr=\sum_{i=1}^{n} r_{s_i}$. Finally, for any integers $X$ and $d$, we denote the value of $X$ modulo $d$ by $[X]_{d} \in \{0,\ldots,d-1\}$.

With these notations, let us define the following Bell  expression:
\begin{equation}\label{eq:GUBI}
\begin{split}
	\hspace{-2mm} \GUBIg &= \sum_{\vec{s} \ {\mathrm{s.t.}} \, \mod{\sums}{m}=0} \, \sum_{r=0}^{k-1} \, \left[r \, - \left\lfloor\frac{\sums}{m}\right\rfloor\right]_k P([\sumr]_k=r) \\[-1mm]
	& \ \ + \!\!\! \sum_{\vec{s} \ {\mathrm{s.t.}} \, \mod{\sums}{m}=1} \, \sum_{r=0}^{k-1} \, \left[-r \, + \left\lfloor\frac{\sums}{m}\right\rfloor\right]_k P([\sumr]_k=r) ,
\end{split}
\end{equation}
where $\lfloor \cdot \rfloor$ is the floor function. $\GUBIg$ is clearly a full-correlation Bell  expression, with the full-correlation function involved being the sum, modulo $k$ of all parties' outputs, i.e., $[\sumr]_k$; furthermore, since for a given choice of settings $\vec s$, it only depends on the sum of inputs $\sums$ and the sum of outputs $\sumr$, $\GUBIg$ is symmetric under any permutation of parties.

The  expression  $\GUBIg$ unifies a few important classes of { full-}correlation Bell  expressions, as illustrated in Fig.~\ref{Fig:GUBI}: for the $(2,m,k)$ case, it reduces to the one appearing in the BKP inequalities (which, in turn, contains both the CGLMP inequalities and the chained Bell inequalities as special cases~\cite{BKP}); for the $(n,2,k)$ case it reduces to the generalized Svetlichny  expression of Ref.~\cite{BBGL} (which contains the  expressions  of Refs.~\cite{collins02,Seevinck} as special cases); for the $(n,3,2)$ case, it reduces to the DIEW of Ref.~\cite{diew}. For details on how these known Bell  expressions  are recovered from $\GUBIg$ (and for an alternative way of writing $\GUBIg$), see Appendix~\ref{App:Reduction}.

\begin{figure}[h!]
 \includegraphics[scale=.08]{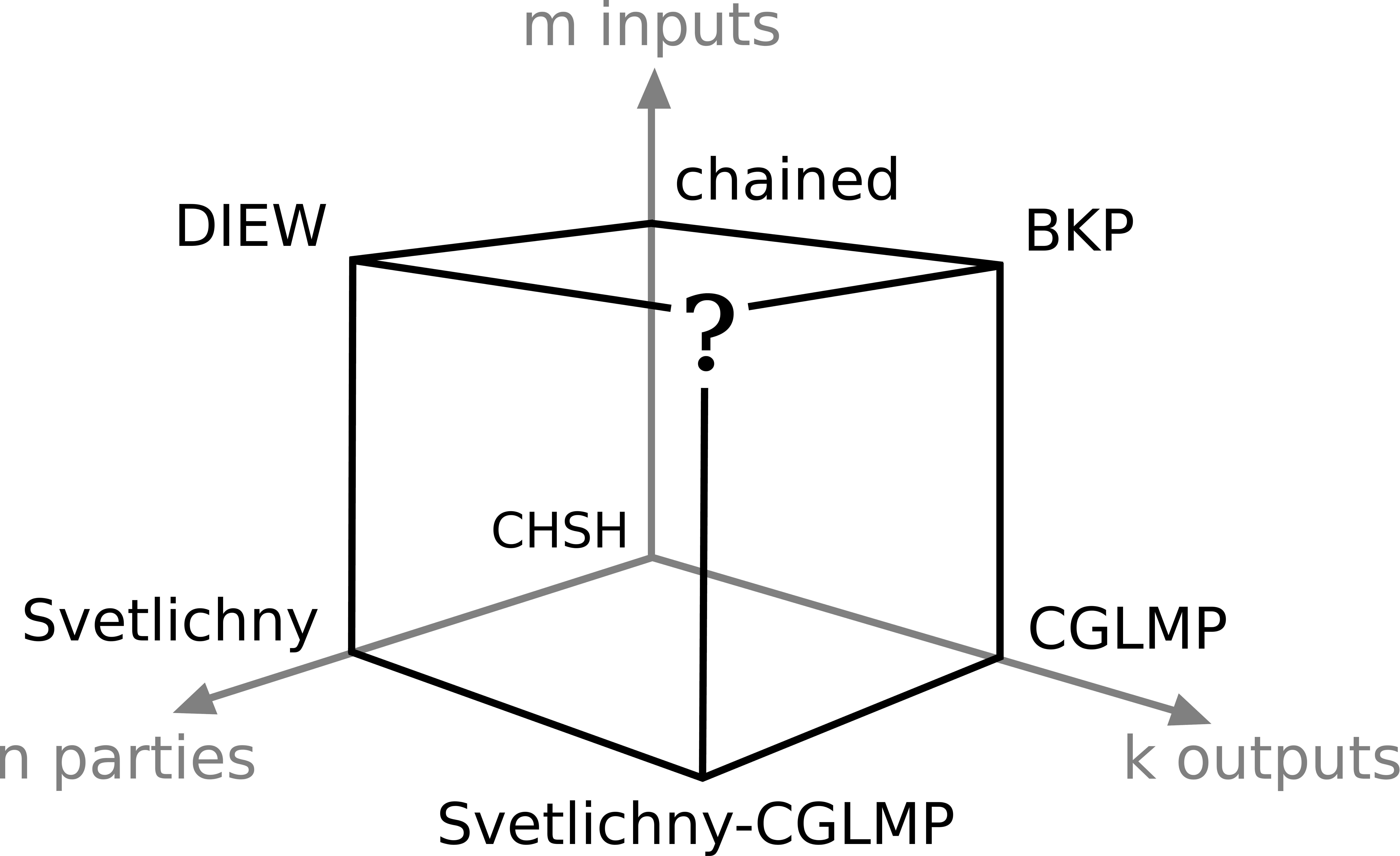}
 \caption{\label{Fig:GUBI} Previously known families of Bell  expressions  are recovered from the $\GUBIg$  expression  of Eq.~(\ref{eq:GUBI}) when $n,m$ or $k$ = 2 (see text); the CHSH expression is recovered, in particular, when $n=m=k=2$. $\GUBIg$ generalizes these  expressions  for $n,m,k>2$, thus completing the vertex ``?'' of the cube depicted above. An even more general full-correlation Bell  expression  is $\GGUBIg$, defined in Eq.~(\ref{eq:GGUBI}) below. }
\end{figure}

\subsection{From Bell  expressions  to Bell-like inequalities}
\label{Sec:Bounds}

Clearly, as it is, $\GUBIg$ is only a linear combination of probabilities. In order to make use of it in, say, entanglement detection, we need to specify the appropriate bounds that depend on the situation of interest.

For example, in a theory where only shared randomness is allowed\footnote{Such theories are also commonly referred to as local (hidden variable) theories.}, one would have
\begin{equation}\label{Eq:GUBI:Local}
	 \GUBIg \stackrel{\L}{\ge} \bound{\L},
\end{equation}
where the {\em local bound} $ \bound{\L}$ is the lower bound of the Bell  expression  admissible within such a theory.\footnote{For simplicity of presentation, we will only discuss the lower bounds on the Bell expressions. Clearly, one can also discuss the upper bounds on $\GUBIg$ analogously.} Here, we have used the symbol ``$\L$'' to remind that the inequality is a constraint that has to be satisfied by a locally causal theory~\cite{bell}; analogous notations will be adopted in all subsequent discussions. 

Inequality~\eqref{Eq:GUBI:Local} is generally called a \emph{Bell inequality}. The CHSH inequality~\cite{Bell:CHSH}, the Pearle-Braunstein-Caves chained inequalities~\cite{chainedBI}, the CGLMP inequalities~\cite{CGLMP}, and the BKP inequalities~\cite{BKP} in Fig.~\ref{Fig:GUBI} are examples of such inequalities that can be written explicitly as:
\begin{subequations}
\begin{gather}
	\GUBI{2}{2}{2} \stackrel{\L}{\ge} 1,\quad
	\GUBI{2}{m}{2} \stackrel{\L}{\ge} 1,\\
	\GUBI{2}{2}{k} \stackrel{\L}{\ge} k-1,\quad
	\GUBI{2}{m}{k} \stackrel{\L}{\ge} k-1.
	\label{Eq:BKP}
\end{gather}
\end{subequations}
The violation of a Bell inequality is a signature of Bell-nonlocality.

Likewise, in a multipartite scenario, one could be interested in detecting genuine multipartite nonlocality (also known as true $n$-body nonseparability~\cite{svet87}). In this case, it is necessary to establish the {\em Svetlichny bound} of $\GUBIg$, which we shall denote by $\bound{\S}$. One can then write down a {\em Svetlichny inequality} in terms of $\GUBIg$ as:
\begin{equation}\label{Eq:GUBI:Svet}
	\GUBIg \stackrel{\S}{\ge} \bound{\S}.
\end{equation}
The inequalities due to Collins {\em et al.}~\cite{collins02} as well as Seevinck-Svetlichny~\cite{Seevinck}, and that presented in Ref.~\cite{BBGL} are inequalities of this kind that can be written explicitly as:
\begin{gather}
	\GUBI{n}{2}{2} \stackrel{\S}{\ge} 2^{n-2}\quad{\rm and}\quad
	\GUBI{n}{2}{k} \stackrel{\S}{\ge} 2^{n-2}(k-1).
\end{gather}
From~\cite{diew}, it also follows that $\GUBI{n}{3}{2} \stackrel{\S}{\ge} 3^{n-2}$.

The quantum violation of a Svetlichny inequality is a sufficient condition for genuine multipartite entanglement. However, for the detection of such entanglement, it already suffices to violate the weaker constraint given by a {\em DIEW}~\cite{diew}, which we can write in the context of $\GUBI{n}{m}{k}$ as:
\begin{equation}\label{Eq:GUBI:BS}
	\GUBIg \stackrel{\B}{\ge} \bound{\B},
\end{equation}
where $\bound{\B}$ is the {\em quantum biseparable bound} on $\GUBIg$. In this notation, the DIEW of Ref.~\cite{diew} can be written as:
\begin{equation}\label{Eq:DIEW:GUBI}
	\GUBI{n}{3}{2} \stackrel{\B}{\ge} 3^{n-2}(3-\sqrt{3}).
\end{equation}

Note finally that for any given scenario $(n,m,k)$, the set of probability distributions allowed in quantum theory is bounded and thus, the Bell expression $\GUBI{n}{m}{k}$ is also restricted in quantum theory to:
\begin{equation}\label{Eq:GUBI:Q}
	 \GUBIg \stackrel{\Q}{\ge} \bound{\Q}.
\end{equation}
Such an inequality is often referred to as a \emph{Tsirelson inequality}, whereas the bound $\bound{\Q}$ is known as a {\em Tsirelson bound}~\cite{B.S.Tsirelson:LMP:1980}. For the CHSH  expression (corresponding to $n{=}m{=}k{=}2$),  for instance , one has
\begin{equation}\label{Eq:Tsirel:CHSH}
	\GUBI{2}{2}{2} \stackrel{\Q}{\ge} 2-\sqrt{2}.
\end{equation}

Note that for general Bell  expressions, these lower bounds obey the following set of inequalities:
\begin{equation}
\begin{split}
	\bound{\Q},\, \bound{\S}\le\bound{\B}\le\bound{\L},
\end{split}
\end{equation}
but the bounds arising from the quantum set and the Svetlichny constraints are not necessarily comparable. For example, three parties that share only  a Popescu-Rohrlich~\cite{PR} box between two of them can clearly generate non-quantum but Svetlichny-local correlations. Conversely, there are Svetlichny inequalities that can be violated quantum mechanically.

\subsection{Generalization to include other Bell expressions}
\label{Sec:GGUBI}

While $\GUBIg$ already embraces a large number of known Bell  expressions, it can actually be further generalized to include an even larger family of Bell expressions. To this end, let us now define:
\begin{equation}\label{eq:GGUBI}
	\GGUBIg = \sum_{\vec{s}} \sum_r \f\!\left([\sums]_m,\left[r-\left\lfloor\frac{\sums}{m}\right\rfloor\right]_k\right) P([\sumr]_k = r) \,,
\end{equation}
where the first sum runs over all possible combinations of settings $\vec{s}$, the second sum runs from $r=0$ to $k-1$, and where $\f:\{0,\ldots,m-1\}\times\{0,\ldots,k-1\}\to\mathbb{R}$ is a real-valued function (defined by $m \times k$ real parameters) that fully characterizes $\GGUBIg$.

As with $\GUBIg$, $\GGUBIg$ is clearly a symmetric, full-correlation Bell  expression. Note the specific form of the arguments of $f(s,r)$, and how the sums $\sums$ and $\sumr$ of inputs and outputs play different roles, with $\sums$ appearing also in the second argument through the quantity $\left\lfloor\frac{\sums}{m}\right\rfloor$. This very term, which is responsible  for the minus sign in the CHSH  expression, turns out to be crucial for the computation of multipartite bounds on $\GGUBIg$ (see next section). 

The Bell expressions $\GUBIg$ introduced previously simply correspond to the choice
\begin{equation}
	\f(\s,\r)= \f_{\cal I}(\s,\r)=
\begin{cases}
\r & \text{if }\s=0 , \\
[-\r]_k & \text{if }\s=1 , \\
0 & \text{otherwise} ,
\end{cases}\ \ \ \label{Eq:f:GUBI}
\end{equation}
so that $\GUBIg = \GGUBIf{n}{m}{k}{f_{\cal I}}$.
Expression~\eqref{eq:GGUBI} is thus more general than~\eqref{eq:GUBI}. However,  not all symmetric full-correlation Bell  expressions  can be put in the form~\eqref{eq:GGUBI}\footnote{For instance, for $n{=}m{=}k{=}2$, the trivial (single-term)  expression  $P(r_{s_1=0} = r_{s_2=0})$ is not of the form~\eqref{eq:GUBI}, since in $\GGUBIf{2}{2}{2}{f}$ the term $P(r_{s_1=1} \neq r_{s_2=1})$ must also come with the same coefficient $f(0,0)$.}.

Still,  the generalized  expression  $\GGUBIg$ now includes some other known families of Bell  expressions  (up to relabellings of inputs and outputs, and possibly affine transformations), such as those appearing in the MABK Bell inequalities~\cite{MerminIneq} and the DIEWs from Ref.~\cite{diewTV}. The parameters leading to these inequalities are summarized in Table~\ref{tbl:KnownBI}.
\begin{center}
\begin{table}[h!]
\begin{tabular}{c|c|c||c||c}
$n$ & $m$ & $\ k \ $ & $\f(\s,\r)$\,\footnote{Notations: $\delta_{\s,0}$ is the Kronecker delta (such that $\delta_{\s,0} = 1$ if $s=0$, $\delta_{\s,0} = 0$ otherwise), and $\Delta$ can be any arbitrary real number.} & Bell  expression  \\ \hline
 $\ge 3$, odd\footnote{Note that the MABK expressions --- often referred to as the Mermin expressions --- are identical to the Svetlichny-Bell  expressions  of Refs.~\cite{collins02,Seevinck} for even $n$. They are thus already recovered by $\GUBI{n}{2}{2}$.} & 2 & 2 & $\delta_{\s,0}\cdot\r$ & MABK~\cite{MerminIneq} \\ \hline
 $\ge 3$ & $\ge 2$ & $ 2$ & $\delta_{\s,0}\cdot\r$ & DIEW~\cite{diewTV} \\ \hline
 $\ge 3$ & $\ge 2$ & $ 2$ & $\cos(\frac{\s-\Delta}{m}\pi)\cdot\r$ & DIEW~\cite{diewTV} \\ \hline
\end{tabular}
\caption{\label{tbl:KnownBI} A summary of some other known Bell  expressions  that can be recovered as special cases of $\GGUBIg$.
}
\end{table}
\end{center}
On top of these, there are a handful of other known bipartite two-output Bell inequalities that are of the form $\GGUBIf{2}{m}{2}{f}$. 
Some of these examples can be found in Eq.~(5) of Ref.~\cite{gisin} and in its Appendix A, as well as in Ref.~\cite{JD:JPA:2010}. 

As mentioned above, in order for Bell  expressions  to be useful in practice, one needs to determine their relevant bounds, so that
\begin{subequations}
\begin{gather}
	\GGUBIg \ \stackrel{\L}{\ge} \ \boundf{\L}{f}, \label{Eq:boundsf_L} \\
	\GGUBIg \ \stackrel{\S}{\ge} \ \boundf{\S}{f}, \label{Eq:boundsf_S} \\
	\GGUBIg \ \stackrel{\B}{\ge} \ \boundf{\B}{f}, \label{Eq:boundsf_B} \\
	\GGUBIg \ \stackrel{\Q}{\ge} \ \boundf{\Q}{f}, \label{Eq:boundsf_Q}
\end{gather}
\end{subequations}
where the various bounds depend on the choice of function $f$. In the next section we show how, starting from bipartite bounds on $\GGUBI{2}{m}{k}$, one can construct bounds on $\GGUBIg$ to obtain multipartite Bell-like inequalities.

\section{From bipartite to multipartite bounds: how to generate new Bell-like inequalities}
\label{Sec:Construction}

Determining the local bound or any of the other bounds described in Sec.~\ref{Sec:Bounds} for a given Bell  expression  is in general a highly nontrivial problem. Nonetheless, we will demonstrate in what follows that once the corresponding Tsirelson and local bounds for the bipartite  expression  $\GGUBI{2}{m}{k}$ are known (for any given choice of $\f$), one can immediately write down, respectively, a Tsirelson inequality and a Svetlichny inequality for $\GGUBIg$ (for the same choice of $\f$). Analogously, we will also demonstrate, in the particular case where $k=2$ and where the function $f$ takes the form $f(s,r) = g(s) \cdot r$, how a quantum biseparable bound on $\GGUBI{n}{m}{2}$ can be obtained by solving a simple optimization problem, for a given $m$ and a given function $g(s)$.

Our starting point is to note that for any $n\geq2$, one can rewrite $\GGUBIg$ as a sum of $m$  expressions  involving effectively one less party. More precisely, let us decompose $\GGUBIg$ as:
\begin{eqnarray}
	&& \hspace{-5mm} \GGUBIg \, = \, \sum_{s_n = 0}^{m-1}\GGUBI{n-1}{m}{k}^{(s_n, r_{s_n})} \,, \label{Eq:S1} \\[1mm]
	&& \hspace{-5mm} \textrm{with} \ \ \GGUBI{n-1}{m}{k}^{(s_n, r_{s_n})} \nonumber \\
	&& =  \!\!\! \sum_{s_1,\ldots,s_{n-1}} \!\!  \sum_r \f\!\left([\sums]_m,\left[r-\left\lfloor\frac{\sums}{m}\right\rfloor\right]_k\right) P([\sumr]_k = r) . \qquad \label{Eq:S2a}
\end{eqnarray}
Defining
\begin{equation}
\begin{array}{rcl}
	s_1^{(s_n)} &=& [s_1+s_n]_m \, = \, s_1+s_n - \left\lfloor\frac{s_1+s_n}{m}\right\rfloor m \,, \\[1mm]
	\sums' &=& s_1^{(s_n)}+\sum_{i=2}^{n-1} s_i \, = \, \sums - \left\lfloor\frac{s_1+s_n}{m}\right\rfloor m \,, \\[2mm]
	r_{s_1}^{(s_n,r_{s_n})} &=& \mod{r_{s_1}+r_{s_n}-\left\lfloor\frac{s_1+s_n}{m}\right\rfloor}{k} \,, \\[1mm]
	\sumrP &=& r_{s_1}^{(s_n,r_{s_n})} + \sum_{i=2}^{n-1} r_{s_i} \,,
\end{array}
\label{eq:substitutions}
\end{equation}
we obtain
\begin{eqnarray}
	&& \hspace{-5mm} \GGUBI{n-1}{m}{k}^{(s_n, r_{s_n})} \nonumber \\
	&& \hspace{-5mm} =  \!\!\! \sum_{s_1^{(s_n)},s_2,\ldots,s_{n-1}} \!\!  \sum_{r'} \f\!\left([\sums']_m,\left[r'-\left\lfloor\frac{\sums'}{m}\right\rfloor\right]_k\right) P([\sumrP]_k = r') . \nonumber \\[-3mm] \label{Eq:S2}
\end{eqnarray}

Thus every term $\GGUBI{n-1}{m}{k}^{(s_n, r_{s_n})}$ appearing in the decomposition~\eqref{Eq:S1} is of the general form~\eqref{eq:GGUBI} for the $n-1$ first parties, and for the same function $f(s,r)$. This implies that, for all given $s_n$ and $r_{s_n}$, $\GGUBI{n-1}{m}{k}^{(s_n, r_{s_n})}$ defines an $(n-1)$-partite Bell expression.

The invariance of $\GGUBIg$ under permutation of the parties implies that the same decomposition can be carried out for any of the other parties. Bearing these in mind, we are now ready to construct some nontrivial multipartite Bell-like inequalities in terms of their bipartite bounds.

\subsection{Tsirelson inequalities}
\label{Sec:Tsirelson}

To derive a Tsirelson inequality for a general multipartite scenario, one can make use of Eq.~\eqref{Eq:S1} recursively and apply inequality~\eqref{Eq:boundsf_Q} for $n=2$. This leads to
\begin{equation}\label{Ineq:Tsirelson:General}
	\GGUBIg \ \stackrel{\tiny \Q}{\ge} \ m^{n-2} \, \boundsf{\Q}{2}{m}{k}{f} \,,
\end{equation}
which is an $n$-partite Tsirelson inequality obtained as a function of the bipartite Tsirelson bound $\boundsf{\Q}{2}{m}{k}{f}$.

\subsection{Svetlichny inequalities}
\label{Sec:Svet}

To derive a Svetlichny bound for the general $(n,m,k)$ scenario, we consider a Svetlichny scenario in which $n-1$ parties are separated into two groups. By hypothesis, cf. Eq.~\eqref{Eq:boundsf_S}, the value of $\GGUBI{n-1}{m}{k}^{(s_n,r_{s_n})}$ for any given value of $s_n$ and $r_{s_n}$ is restricted by the Svetlichny bound $\boundsf{\S}{n-1}{m}{k}{f}$.
Let us then introduce a new party (labeled by ``$n$''), and (without loss of generality\footnote{This follows from the possibility to perform analogous decomposition as in Eqs.~(\ref{Eq:S1};\ref{Eq:S2}) for any other party.}) let it join the same group as the first party; this does not change the total number of groups. 
Eq.~\eqref{eq:substitutions} and Eq.~\eqref{Eq:S2} can then be interpreted as follows: since the first and the $n^{\text{th}}$ parties are in the same group, they can collaborate and thus the $n^{\text{th}}$ party can communicate to the first party his/her input and output (and vice versa). The first party can thus define new effective inputs $s_1^{(s_n)}$ and outputs  $r_{s_1}^{(s_n,r_{s_n})}$ as in~Eq.~\eqref{eq:substitutions}:  this allows the first party to account for every possible strategy of the new party. We thus see that in this new scenario, we must also have\footnote{If the new party does not join any of the existing groups, the $n$ parties can clearly only do worse in terms of minimizing $\GGUBI{n-1}{m}{k}^{(s_n, r_{s_n})}$.}
\begin{equation}
	 \GGUBI{n-1}{m}{k}^{(s_n, r_{s_n})} \ \stackrel{\tiny \S}{\ge} \ \boundsf{\S}{n-1}{m}{k}{f}.
\end{equation}
By repeating the above argument recursively and noting that $\boundsf{\S}{2}{m}{k}{f}=\boundsf{\L}{2}{m}{k}{f}$, i.e., that the Svetlichny and local bounds coincide for $n=2$, we thus obtain the Svetlichny inequalities:
\begin{equation}\label{Eq:Svet:Final}
	\GGUBIg \ \stackrel{\tiny \S}{\ge} \ m^{n-2} \, \boundsf{\L}{2}{m}{k}{f}.
\end{equation}

Note that the bound corresponding to the situation in which the $n$ parties are separated into $\textsc{g}$ groups~\cite{JD:QuantifyNonlocality} can be derived in a similar way, from the local bound of $\GGUBI{\textsc{g}}{m}{k}$.

\subsection{Two-output DIEWs}
\label{Sec:DIEW}

Consider now the case where the outputs are binary ($k=2$), and $\f(\s,\r)=\g(\s)\cdot\r$ for some function $\g:\{0,\ldots,m-1\}\to\mathbb{R}$ (as in the examples of Table~\ref{tbl:KnownBI} for instance). The probabilities $P([\sumr]_2=r)$ appearing in $\GGUBIf{n}{m}{2}{g\cdot r}$
can in this case be expressed in terms of the commonly used $n$-partite correlators\footnote{The correlator $E_{\vec{s}}$ can be seen as the average value of the product of experimental outcomes, when these are labeled by $\pm1$.} $E_{\vec{s}}= P([\sumr]_2=0) - P([\sumr]_2=1)$, so that
\begin{equation}
	P([\sumr]_2=r) = \half\left[1+(-1)^r E_{\vec{s}}\right] \,.
\end{equation}
We then obtain
\begin{eqnarray}
&& \hspace{-8mm} \GGUBIf{n}{m}{2}{g.r} \ = \ \half \sum_{\vec{s}} g\!\left([\sums]_m\right) \left[1-(-1)^{\left\lfloor\frac{\sums}{m}\right\rfloor} E_{\vec{s}}\right] \nonumber\\
&& \hspace{-5mm} = \ \half \left[ m^{n-1} \! \sum_{s=0}^{m-1} g(s) \, - \, \sum_{\vec{s}} g\!\left([\sums]_m\right) (-1)^{\left\lfloor\frac{\sums}{m}\right\rfloor} E_{\vec{s}} \right] \!, \ \label{Eq:TwoOutput:AntiPeriodic}
\end{eqnarray}
where we used the fact that for each value of $s = 0,\ldots,m-1$, there are $m^{n-1}$ lists of settings $\vec s$ such that $[\sums]_m = s$.
Any lower bound $\boundsf{}{n}{m}{2}{g.r}$ on $\GGUBIf{n}{m}{2}{g.r}$ will thus be related to a corresponding upper bound on the last sum of Eq.~\eqref{Eq:TwoOutput:AntiPeriodic} by an affine transformation.

In the case of biseparability in particular, we show in Appendix~\ref{App:Biseparable} how to determine the biseparable bound on Eq.~\eqref{Eq:TwoOutput:AntiPeriodic} for $n=3$. A biseparable bound for general $n$ can then be derived straightforwardly by invoking the recursive arguments employed in the previous subsections. This thus allows us to obtain, from Eq.~\eqref{DIEW_3m2}, the following two-output DIEWs:
\begin{equation}\label{Eq:DIEW}
\begin{split}
	& \GGUBIf{n}{m}{2}{g.r} \ \stackrel{\B}{\ge} \\
	& \frac{1}{2}m^{n-2} \! \left( \! m \! \sum_{s=0}^{m-1} \!\g(s) - \!\! \max_{j=0,\ldots,m-1} \! \left[ \eta_{j}
	\csc\frac{\eta_{j}\pi}{2m} \Big| \!\sum_{\s=0}^{m-1} \g(\s) \, \omega_j^{\s} \Big| \right] \right),
\end{split}
\end{equation}
where $\eta_{j}$ is the greatest common divisor of $2j+1$ and $m$, while $\omega_j=e^{i\frac{\pi}{m}\left(2j+1\right)}$.

\subsection{Three explicit examples}

We showed in the previous subsections how to derive multipartite bounds on the general  expression  $\GGUBIg$, from bi- or tri-partite bounds. Applying the above results to the more specific case of $\GUBIg =\GGUBIf{n}{m}{k}{f_{\cal I}}$, cf. Eq.~\eqref{Eq:f:GUBI}, we now derive three explicit examples of new Bell-like inequalities.

\medskip

1. For the  expression \blk  $\GUBIg$, we have the bipartite local bound $\bounds{\L}{2}{m}{k}=k-1$, cf. Eq.~\eqref{Eq:BKP}. Substituting this into Eq.~\eqref{Eq:Svet:Final}, we thus obtain the following Svetlichny inequality for arbitrary numbers of parties, inputs and outputs:
\begin{equation}\label{Ineq:SvetNMK}
	 \GUBIg \ \stackrel{\S}{\ge} \ m^{n-2} (k-1).
\end{equation}
The case $m=2$ of this  expression, previously derived in Ref.~\cite{BBGL}, is marked as Svetlichny-CGLMP in Fig.~\ref{Fig:GUBI}. Inequality~\eqref{Ineq:SvetNMK} represents the Svetlichny inequality for the vertex marked by ``?'' in the cube shown in Fig.~\ref{Fig:GUBI}.

\medskip

2. For binary outputs ($k=2$), since $[-r]_2 = [r]_2$, the function $f_{\cal I}(s,r)$ specified in Eq.~\eqref{Eq:f:GUBI} is of the form $f_{\cal I}(s,r)=g_{\cal I}(s) \cdot r$, with $g_{\cal I}(s) = 1$ if $s=0$ or 1, and $g_{\cal I}(s) = 0$ otherwise. For this choice, one gets $\sum_{s=0}^{m-1} \g_{\cal I}(s) = 2$ and $\left|\sum_{\s=0}^{m-1} \g_{\cal I}(\s) \, \omega_j^{\s}\right| = 2\left|\cos\tfrac{(2j+1)\pi}{2m}\right|$. Substituting these into Eq.~\eqref{Eq:DIEW} and after some computation\footnote{From~\eqref{Eq:DIEW}, one needs to calculate $\max_{j} \! \left[ \eta_{j} \csc\frac{\eta_{j}\pi}{2m} \left|\cos\tfrac{(2j+1)\pi}{2m}\right| \right]$. By decomposing (when $\eta_j \neq m$) the  expression  to maximize in the form $\left|\eta_{j} \cot\frac{\eta_{j}\pi}{2m}\right| \cdot \left|\frac{\cos\tfrac{(2j+1)\pi}{2m}}{\cos\frac{\eta_{j}\pi}{2m}}\right|$, the first absolute value is maximized for $\eta_{j}$ as small as possible, while the second one is upper bounded by 1. The maximum one needs to calculate is thus found to be $\cot\frac{\pi}{2m}$, obtained for $j=0$.}, one arrives at the following two-output DIEWs for arbitrary numbers of parties and inputs:
\begin{equation}\label{Ineq:DIEWNew}
	 \GUBI{n}{m}{2} \ \stackrel{\B}{\ge} m^{n-2}\left(m-\cot\frac{\pi}{2m}\right).
\end{equation}

\medskip

3. In a similar manner, it follows from the result of Ref.~\cite{S.Wehner:PRA:022110} that the Tsirelson bound for $\GUBI{2}{m}{2}$ is $\bounds{\Q}{2}{m}{2}=m\left(1-\cos\frac{\pi}{2m}\right)$. Substituting this into Eq.~\eqref{Ineq:Tsirelson:General} then gives the following $n$-partite, $m$-setting Tsirelson inequality:
\begin{equation}
	 \GUBI{n}{m}{2} \ \stackrel{\Q}{\ge} m^{n-1}\left(1-\cos\frac{\pi}{2m}\right).
\end{equation}

\subsection{Tightness of our inequalities}

Evidently, it is desirable to understand if the Tsirelson inequalities, Svetlichny inequalities and DIEWs derived using the above procedures can be saturated. Notice that a key common feature in these derivations involves Eq.~\eqref{Eq:S1}. Hence, the $n$-partite bound can be saturated only if all the $(n-1)$-partite bounds on the  expressions  $\GGUBI{n-1}{m}{k}^{(s_n, r_{s_n})}$ involved in Eq.~\eqref{Eq:S1} can be simultaneously saturated.

In general, one may thus expect that the $n$-partite bounds and hence the inequalities derived in Sections~\ref{Sec:Tsirelson}~--~\ref{Sec:DIEW} are not necessarily tight. Nonetheless, for all the examples that we have checked, all these bounds can indeed be saturated. For example, for the  expression  $\GUBIg$, both the Svetlichny bound $\bounds{\S}{n}{3}{2}=3^{n-2}$ and the biseparable bound $\bounds{\B}{n}{3}{2}=3^{n-2}(3-\sqrt{3})$ obtained above can be saturated~\cite{diew}; likewise, it can be verified\footnote{This can be done, for example, using the optimization tools of Ref.~\cite{Liang07} and the converging hierarchy of semidefinite programs discussed in Ref.~\cite{sdp-hierarchy}.} that the Tsirelson bounds for $\GUBIg$ satisfy $\bounds{\Q}{4}{2}{2}=2\bounds{\Q}{3}{2}{2}=4\bounds{\Q}{2}{2}{2} = 4 \cdot (2-\sqrt{2})$ whereas $\bounds{\Q}{4}{2}{3}\approx2\bounds{\Q}{3}{2}{3}\approx4\bounds{\Q}{2}{2}{3} \approx 4\times\left(3-\sqrt{\tfrac{11}{3}}\right)$.\footnote{These last set of equalities were only verified numerically, up to the numerical precision of $10^{-9}$.}

\section{Concluding remarks}
\label{Sec:Conclusion}

Starting from a unified Bell  expression  $\GUBIg$, we have shown that various important correlation Bell  expressions  can be recovered as special cases (cf. Fig.~\ref{Fig:GUBI}). A natural generalization of $\GUBIg$ to $\GGUBIg$ has, in turn, allowed us to also recover other known correlation Bell  expressions  that have previously been investigated in the literature.

Within the framework of $\GGUBIg$, we also demonstrated how multipartite Tsirelson inequalities, Svetlichny inequalities and device-independent witnesses for genuine multipartite entanglement (DIEWs) can be constructed. This, in particular, has allowed us to construct a new family of Svetlichny inequalities for arbitrary numbers of parties, inputs and outputs as well as a new family of two-output DIEWs that can be applied to a scenario involving arbitrary numbers of parties and inputs.

Clearly, a natural question that one may ask is how useful the (new) inequalities that can be constructed within this framework are. To this end, we note that inequality~\eqref{Ineq:SvetNMK} has recently also been discovered independently by Aolita {\em et al.}~\cite{Aolita:1109.3163} and used to show that the higher-dimensional $n$-partite Greenberger-Horne-Zeilinger (GHZ) states can exhibit fully random genuinely multipartite quantum correlations. 

The DIEW given in inequality~\eqref{Ineq:DIEWNew}, on the other hand, can also be shown to detect the genuine multipartite entanglement of a noisy GHZ state up to the same level of noise resistance (visibility) --- for any given $m$ and $n$ --- as those given in Ref.~\cite{diewTV}.  Finally, it is worth noting that the existing techniques for computing Tsirelson bounds (such as those discussed in Ref.~\cite{sdp-hierarchy}) generally do not work very well beyond small values of $n$ and/or $k$. Our general Tsirelson inequality~\eqref{Ineq:Tsirelson:General} may thus serve as a useful tool for characterizing and understanding the extent of nonlocality allowed in quantum theory. We believe that our inequalities and the framework from which they were constructed, given their generality and simplicity, have the potential for many other interesting applications.

Evidently, there are many open problems that stem from the present work. An obvious question that we have not addressed, for instance, is whether there is any choice of the function $f(s,r)$ for which the local bound on $\GGUBIg$ can be easily determined, and whether the resulting inequalities correspond to facets of the respective local polytopes.

As we already acknowledged, the framework that we have provided does not allow one to consider all possible full-correlation Bell-like inequalities. The  expression  $\GGUBIg$ defined in Eq.~(\ref{eq:GGUBI}) was constructed so that it has the nice property of being decomposable as in~(\ref{Eq:S1}), namely, as a sum of $m$ $(n{-}1)$-partite Bell  expressions  of the same form; there are however symmetric full-correlation  expressions  which cannot be written in such a way (see, e.g. Ref.~\cite{Liang10}). Besides, it could also be interesting to look at correlation Bell-like inequalities that do not have full symmetry with respect to permutation of parties. We shall leave these possibilities for future research.

\begin{acknowledgments}
We acknowledge useful discussions with Tam\'as V\'ertesi, Stefano Pironio and Antonio Ac\'in. This work is supported by the UK EPSRC, a UQ Postdoctoral Research Fellowship, the Swiss NCCR ``'Quantum Photonics'',  the Spanish MICINN through CHIST-ERA DIQIP, and the European ERC-AG QORE.
\end{acknowledgments}

\appendix

\section{Reduction of $\GUBIg$ to known Bell  expressions }
\label{App:Reduction}

In this Appendix, we show that by appropriate relabelling of inputs and outputs, and possibly by applying some affine transformation, $\GUBIg$ reduces to the respective Bell  expressions  given in Fig.~\ref{Fig:GUBI}.

We start by noting that $\GUBIg$ can alternatively be written using the bracket notation introduced in Ref.~\cite{Acin06} via the average values $\langle R \rangle = \sum_{r=0}^{k-1} r P(R=r)$:
\begin{equation}\label{eq:GUBI2}
\begin{split}
	 & \hspace{-2mm} \GUBIg \\
	 & = \! \sum_{\stackrel{\vec{s} \ {\mathrm{s.t.}}}{\mod{\sums}{m}=0}} \!\! \left\langle \left[\sumr - \! \left\lfloor\frac{\sums}{m}\right\rfloor\right]_k \right\rangle + \!\! \sum_{\stackrel{\vec{s} \ {\mathrm{s.t.}}}{\mod{\sums}{m}=1}} \!\! \left\langle \left[-\sumr + \! \left\lfloor\frac{\sums}{m}\right\rfloor\right]_k \right\rangle .
\end{split}
\end{equation}

\subsection{Reduction to known two-party Bell  expressions }
\label{Sec:BKP}

For the case of $n=2$, $\GUBIg$ simplifies to 
\begin{equation}\label{Eq:Omega:2mk}
\begin{split}
\GUBI{2}{m}{k} &= \sum_{\mod{x+y}{m}=0} \left\langle \left[ A_x+B_y -\left\lfloor\frac{x+y}{m}\right\rfloor\right]_k \right\rangle \\
			 &\quad+ \sum_{\mod{x+y}{m}=1} \left\langle \left[-A_x-B_y +\left\lfloor\frac{x+y}{m}\right\rfloor\right]_k \right\rangle,
\end{split}
\end{equation} 
where for ease of comparison with the notation adopted in Ref.~\cite{BKP}, we have written $s_1=x$, $s_2=y$, $r_{s_1}=A_x$ and $r_{s_2}=B_y$. Introducing the new output variables $B_0' = [-B_0]_k$ and $B_y' = [1-B_{m-y}]_k$ for $y \geq 1$, the above expression becomes
\begin{eqnarray}
\GUBI{2}{m}{k} &\ = \ & \left\langle[A_0+B_0]_k\right\rangle+\sum_{x=1}^{m-1}\left\langle\left[A_{x}+B_{m-x}-1 \right]_k\right\rangle \nonumber \\
&& +\left\langle[-A_0-B_1]_k\right\rangle+\left\langle[-A_1-B_0]_k\right\rangle \nonumber \\
&& \qquad +\sum_{x=2}^{m-1}\left\langle\left[-A_x-B_{m+1-x}+1\right]_k\right\rangle, \\
&\ = \ & \sum_{x=0}^{m-1}\left\langle\left[A_{x}-B_{x}' \right]_k\right\rangle +\sum_{x=1}^{m-1}\left\langle\left[B_{x-1}'-A_x\right]_k\right\rangle \nonumber \\
&& \qquad \qquad +\left\langle[B_{m-1}'-A_0-1]_k\right\rangle ,
\end{eqnarray} 
which is precisely the Bell expression due to BKP~\cite{BKP}.

\subsection{Reduction to known two-input Bell  expressions }
\label{Sec:SvetCGLMP}

For the case with two inputs, i.e., $m=2$, all terms with all possible combinations of inputs appear in $\GUBI{n}{2}{k}$:
\begin{equation}
	\GUBI{n}{2}{k} \, = \sum_{\vec{s}} \ \left\langle \left[(-1)^\sums\left( \sumr \, - \left\lfloor\frac{\sums}{2}\right\rfloor\right)\right]_k \right\rangle \,.
\end{equation}
Defining the new output variables $r_{s_1}'=\mod{r_{s_1}-s_1+1}{k}$ and $r_{s_i}'=\mod{r_{s_i}-s_i}{k}$ for all $i=2,\ldots,n$, as well as the new sum $\sumr'=\sum_{i=1}^{n} r_{s_i}'$, so that $[\sumr - \left\lfloor\frac{\sums}{2}\right\rfloor]_k = [\sumr' + \sums - 1 - \left\lfloor\frac{\sums}{2}\right\rfloor]_k = [\sumr' + \left\lfloor\frac{\sums-1}{2}\right\rfloor]_k$, we can rewrite $\GUBI{n}{2}{k}$ as:
\begin{eqnarray}
	\GUBI{n}{2}{k} &\ = \ & \sum_{\vec{s}} \ \left\langle \left[(-1)^\sums\left( \sumr' + \left\lfloor\frac{\sums-1}{2}\right\rfloor\right)\right]_k \right\rangle \,, \quad
\end{eqnarray}
which is precisely, up to a change of notations ($r_{s_j=0}' \leftrightarrow a_j', r_{s_j=1}' \leftrightarrow a_j$), the $n$-partite Svetlichny-CGLMP Bell  expression  of Ref.~\cite{BBGL} (see also Ref.~\cite{svet:chinese}): one can indeed check that $\GUBI{2}{2}{k}$ is the same as $S_{2,d}$ in Eq.~(6) of~\cite{BBGL}, and that $\GUBI{n}{2}{k}$ satisfies the recursive rules of Eq.~(7) and Eq.~(9) of~\cite{BBGL} (note that the terms with $(-1)^\sums = 1$, respectively $-1$, in the sum above correspond to the terms denoted as $[\ldots]$ and $[\ldots]^*$ in~\cite{BBGL}).

\subsection{Reduction to known two-output Bell  expressions }

In the case of binary outputs, by applying Eq.~\eqref{Eq:TwoOutput:AntiPeriodic} to $\GUBI{n}{m}{2}$, one finds that $\GUBI{n}{m}{2}$ is equivalent to $\sum_{\vec{s}, [\sums]_m=0,1} (-1)^{\left\lfloor\frac{\sums}{m}\right\rfloor} E_{\vec{s}}$. For $m=3$, this is precisely the DIEW introduced in Ref.~\cite{diew}.

More generally, for $k=2$ and when $f(s,r)$ is of the form $g(s) \cdot r$, one finds from Eq.~\eqref{Eq:TwoOutput:AntiPeriodic} that $\GGUBIf{n}{m}{2}{g.r}$ is equivalent to a symmetric full-correlation Bell expression which is characterized by coefficients of the form $\g([\sums]_m) (-1)^{\lfloor\frac{\sums}{m}\rfloor}$, i.e., a discrete function of $\sums$ that is {\em antiperiodic} with antiperiod $m$. This is a characteristic shared by several previously known Bell  expressions; in particular, the MABK Bell inequalities~\cite{MerminIneq}, and the DIEWs discussed in Ref.~\cite{diewTV} can also be recovered from $\GGUBIf{n}{m}{k}{g.r}$ (see Table~\ref{tbl:KnownBI}).

\section{Computing the tripartite biseparable bound of Eq.~(\ref{Eq:TwoOutput:AntiPeriodic})}
\label{App:Biseparable}

From the definition of a biseparable bound, it follows that the quantum biseparable upper bound on the last sum of Eq.~\eqref{Eq:TwoOutput:AntiPeriodic} can be written explicitly as (cf. Appendix~B in the Suppl. Mat. of Ref.~\cite{diew} 
\begin{equation}\label{Eq:BSBound}
	\max_{A_x=\pm1}	\max_\rho
	\sum_{x,y,z} \g([x{+}y{+}z]_m) (-1)^{\lfloor\frac{x{+}y{+}z}{m}\rfloor} A_x \langle \hat B_y\otimes \hat C_z \rangle_\rho \,,
\end{equation}
where $\rho$ is any quantum state shared by two parties, Bob and Charlie, and $\hat B_y$, $\hat C_z$ are quantum observables that satisfy $\hat B_y^2=\unit$, $\hat C_z^2=\unit$. That is, the required biseparable bound is the Tsirelson bound for a bipartite Bell inequality between Bob and Charlie with coefficients defined by 
\begin{equation}\label{Eq:EffectiveBipartiteCoefficeints}
	M_{yz}^{\{\!A_x\!\}} = \sum_{x=0}^{m-1} \g([x{+}y{+}z]_m) (-1)^{\lfloor\frac{x{+}y{+}z}{m}\rfloor} A_x,
\end{equation}
but further maximized over all possible choices of the third party (Alice's) strategies $A_x=\pm 1$. It thus follows that a (not necessarily tight) biseparable bound on Eq.~\eqref{Eq:TwoOutput:AntiPeriodic} can be obtained by solving the semidefinite program formalized in Ref.~\cite{S.Wehner:PRA:022110} and optimizing over the choices of $A_x$.

To this end, let us follow Ref.~\cite{diewTV} and construct a $m \times m$ matrix $M^{\{\!A_x\!\}}$ with coefficients given by Eq.~\eqref{Eq:EffectiveBipartiteCoefficeints}, but with $z$ replaced\footnote{This corresponds to a relabelling of the input of Charlie, which does not change the biseparable bound of the Bell expression.} by $m{-}1{-}z$ (here, $y$ and $z$ represent, respectively, the row and column indices of $M^{\{\!A_x\!\}}$). By the weak duality of semidefinite programs~\cite{sdp}, and from the results of Ref.~\cite{S.Wehner:PRA:022110}, it can be shown that an upper bound on the Tsirelson bound of any bipartite, $m$-input, 2-output, Bell correlation inequality with coefficients defined by $M^{\{\!A_x\!\}}$ is given by $m$ times the {\em largest singular value} of $M^{\{\!A_x\!\}}$.

Note that the matrix $M^{\{\!A_x\!\}}$ thus constructed from Eq.~\eqref{Eq:EffectiveBipartiteCoefficeints} is a Toeplitz matrix (more precisely, a ``modified circulant matrix''~\cite{diewTV,Toepliz}) with (orthogonal) eigenvectors $v_j=\left(1,\omega_j,\ldots, \omega_j^{m-1}\right)$ and corresponding eigenvalues
\begin{equation}\label{Eq:sj}
	\lambda_j^{\{\!A_x\!\}} = \sum_{x=0}^{m-1} A_x\omega_j^{x}\,\sum_{\s=0}^{m-1} \g(\s)\omega_j^{m-1-\s} \,,
\end{equation}
for $\omega_j=e^{i\frac{\pi}{m}\left(2j+1\right)}$.
Furthermore, one can show that $M^{\{\!A_x\!\}}$ is normal, and therefore its singular values are given by the absolute values of its eigenvalues. The desired biseparable bound can then be obtained by computing $\max_{j=0,\ldots,m-1}\max_{A_x=\pm1} \left|\lambda_j^{\{\!A_x\!\}}\right|$, which can be achieved using the following Lemma.

\begin{lemma}\label{Lemma:MaxOverA}
For a given integer $j = 0,\ldots,m{-}1$, let $\eta_{j}$ be the greatest common divisor of $2j+1$ and $m$; then 
\begin{equation}\label{Eq:MaxOverA}
	\max_{A_x=\pm1} \left|\sum_{x=0}^{m-1} A_x\omega_j^{x}\right|
	= \eta_{j} \, \csc\frac{\eta_{j}\pi}{2m}.
\end{equation}
\end{lemma}
\begin{proof}
To prove this, first note that each $\omega_j^x$ is a $2m$-root of unity and can therefore be understood as a phase vector on the complex plane. The above optimization over $A_x$ is thus simply a maximization of the magnitude of (the vectorial sum) $\sum_x A_x\omega_j^{x}$, which can be achieved by concentrating $A_x\omega_j^x$ as much as possible, at most on half a plane. Hence, an {\em optimal} choice of $A_x$ corresponds to setting $A_x=1$ when the argument of $\omega_j^x$ is in $[0,\pi)$, and $A_x=-1$ when its argument is in $[\pi,2\pi)$.

It then follows from the definition of $\omega_j$ that $A_x\omega_j^x=\omega_0^{\ell_{j,x}}$ where $\ell_{j,x}= (2j{+}1)x\,\, \text{mod}\,\, m$. Moreover, as $x$ increases from 0 to $m-1$ in steps of 1, the integer $\ell_{j,x}$ \emph{is never repeated} until $x$ hits $\tfrac{m}{\eta_{j}}$, in which case $\ell_{j,x}=\ell_{j,0}=0$. Next, note that $2j{+}1$ and $m$ are both integer multiples of $\eta_{j}$, it thus follows that $\ell_{j,x}$ must also be an \emph{integer multiple} of $\eta_{j}$. This, together with the fact that there are $\tfrac{m}{\eta_{j}}$ distinct values of $\ell_{j,x}$ as $x$ varies from 0 to $\tfrac{m}{\eta_{j}}-1$ implies that we must have
\begin{equation}
	\{\ell_{j,x}\}_{x=0}^{\tfrac{m}{\eta_{j}}-1}=\{0,\eta_{j},2\eta_{j},\ldots, m{-}\eta_{j}\} = \{x\eta_j\}_{x=0}^{\tfrac{m}{\eta_{j}}-1}.
\end{equation}	
Geometrically, this means that all neighboring phase vectors in the set $\{A_x\omega_j^x\}_{x=0}^{\tfrac{m}{\eta_{j}}-1}$ are equally spaced.

Finally, note that because $A_{\tfrac{m}{\eta_{j}}}\omega_j^{\tfrac{m}{\eta_{j}}}=1$, the phase vectors $A_x\omega_j^x$ for larger values of $x$ will be identical to those with $0\le x\le \tfrac{m}{\eta_{j}}$. Bearing all these in mind, the left-hand-side of Eq.~\eqref{Eq:MaxOverA} can now be evaluated to give
\begin{align}
	\max_{A_x=\pm1} \left|\sum_{x=0}^{m-1} A_x\omega_j^{x}\right| = \eta_{j} \left|\sum_{x=0}^{\frac{m}{\eta_{j}}-1} \omega_0^{x\eta_{j}}\right| = \eta_{j} \csc\frac{\eta_{j}\pi}{2m},
\end{align}
thus completing the proof of Lemma~\ref{Lemma:MaxOverA}.
\end{proof}

Putting all these together, and noting that $\left| \sum_{\s=0}^{m-1} \g(\s) \, \omega_j^{m-1-\s} \right| = \left| \sum_{\s=0}^{m-1} \g(\s) \, \omega_j^{\s} \right|$, we thus see that for $n=3$, the last sum in Eq.~\eqref{Eq:TwoOutput:AntiPeriodic} admits a biseparable (upper) bound of:
\begin{equation}
	m \times \max_j \left[ \eta_{j} \, \csc\frac{\eta_{j}\pi}{2m} \, \Big| \sum_{\s=0}^{m-1} \g(\s) \, \omega_j^{\s} \Big| \right],
\end{equation}
implying
\begin{equation} \label{DIEW_3m2}
\begin{split}
	& \GGUBIf{3}{m}{2}{g.r} \ \stackrel{\B}{\ge} \\
	& \frac{1}{2} m \! \left( m \! \sum_{s=0}^{m-1} \!\g(s) - \max_j \! \left[ \eta_{j}
	\csc\frac{\eta_{j}\pi}{2m} \Big| \!\sum_{\s=0}^{m-1} \g(\s) \, \omega_j^{\s} \Big| \right] \right). \quad
\end{split}
\end{equation}

\end{document}